\documentclass[11pt, a4paper, reqno, 11pt]{article}
\usepackage{a4wide}
\usepackage{amsmath}
\usepackage{amssymb}
\usepackage{amsthm} 
\usepackage{hyperref}
\hypersetup{colorlinks=true,citecolor=blue,linkcolor=blue,urlcolor=blue}
\usepackage{tikz}

\newtheorem{theorem}{Theorem}
\newtheorem{proposition}[theorem]{Proposition}
\theoremstyle{definition}
\newtheorem{definition}[theorem]{Definition}
\newtheorem{example}[theorem]{Example}
\newtheorem{remark}[theorem]{Remark}

\newcommand{\I}{I}
\renewcommand{\deg}{\operatorname{deg}}
\newcommand{\Par}{\textsf{Par}}
\newcommand{\CSP}{\textnormal{CSP}}
\newcommand{\CSPP}{\textnormal{CSP($P$)}}
\newcommand{\NAE}{\textsf{NAE}}
\newcommand{\Cut}{\textsf{Cut}}
\DeclareMathOperator{\Val}{Val}
\newcommand{\eps}{\varepsilon}
\newcommand{\RR}{\mathbb{R}_+}
\newcommand{\tuple}[1]{\ensuremath{(#1)}}

\newcommand\Crestrict[2]{{
  \left.\kern-\nulldelimiterspace 
  #1 
  \right|_{#2} 
  }}

\begin{document}

\title{Sparsification of Binary CSPs\thanks{An extended abstract of this work
appeared in \emph{Proceedings of the 36th International Symposium on
Theoretical Aspects of Computer Science (STACS 2019)}~\cite{bz19:stacs}. 
Stanislav \v{Z}ivn\'y was supported by a Royal Society University Research
Fellowship. Work mostly done while Silvia Butti was at the University of Oxford.
The project that gave rise to these results received the support of a fellowship
from ``a Caixa'' Foundation (ID 100010434). The fellowship code is
LCF/BQ/DI18/11660056. This project has received funding from the European
Union’s Horizon 2020 research and innovation programme grant agreement No 714530
and under the Marie Sk\l{}odowska-Curie grant agreement No 713673. The paper
reflects only the authors' views and not the views of the ERC or the European
Commission. The European Union is not liable for any use that may be made of the
information contained therein.}}

\author{
Silvia Butti\\
Department of Information and Communication Technologies\\
Universitat Pompeu Fabra, Spain\\
\texttt{silvia.butti@upf.edu}
\and
Stanislav \v{Z}ivn\'y\\
Department of Computer Science\\
University of Oxford, UK\\
\texttt{standa.zivny@cs.ox.ac.uk}
}

\date{}
\maketitle

\begin{abstract}

  A cut $\eps$-sparsifier of a weighted graph $G$ is a re-weighted subgraph of
  $G$ of  (quasi)linear size that preserves the size of all cuts up to a
  multiplicative factor of $\eps$. Since their introduction by Bencz\'ur and
  Karger [STOC'96], cut sparsifiers have proved extremely influential and found
  various applications. Going beyond cut sparsifiers, Filtser and Krauthgamer
  [SIDMA'17] gave a precise classification of which binary Boolean CSPs are
  sparsifiable. In this paper, we extend their result to binary CSPs on
  arbitrary finite domains.

\end{abstract}

\section{Introduction}
\label{sec:intro}

The pioneering work of Bencz\'ur and Karger~\cite{Benc} showed that every
edge-weighted undirected graph $G=(V,E,w)$ admits a cut-sparsifier. In
particular, assuming that the edge weights are positive, for every $0<\eps<1$
there exists (and in fact can be found efficiently) a re-weighted subgraph
$G_\eps=(V,E_\eps\subseteq E,w_\eps)$ of $G$ with $|E_\eps|=O(\eps^{-2}n\log n)$
edges such that 
\begin{equation*}
  \forall S\subseteq V, \qquad \Cut_{G_{\eps}}(S)\in(1\pm\eps)\Cut_G(S),
\end{equation*}
where $n=|V|$ and $\Cut_G(S)$
denotes the total weight of edges in $G$ with exactly one endpoint in $S$. The
bound on the number of edges was later improved to $O(\eps^{-2}n)$ by Batson,
Spielman, and Srivastava~\cite{Bats}. Moreover, the bound $O(\eps^{-2}n)$ is
known to be tight by the work of Andoni, Chen, Krauthgamer, Qin, Woodruff, and
Zhang~\cite{And}.

The original motivation for cut sparsification was to speed up algorithms for
cut problems and graph problems more generally. The idea turned out to be very
influential, with several generalisations and extensions, including, for
instance, sketching~\cite{Ahn09:icalp,And}, sparsifiers for cuts in
hypergraphs~\cite{Kog,Newman13:sicomp}, and spectral
sparsification~\cite{Spi,Spielman04:stoc,Spielman11:sicomp-graph,Fung11:stoc,Soma19:soda}.

Filtser and Krauthgamer~\cite{Fil} considered the following natural question: which binary
Boolean CSPs are sparsifiable? In order to state their results as well as our
new results, we will now define binary constraint satisfaction problems. 

An instance of the binary\footnote{Some papers use the term
\emph{two-variable}.} \emph{constraint satisfaction problem} (CSP) is a
quadruple $\I=(V,D,\Pi,w)$, where $V$ is a set of variables, $D$ is a finite set
called the \emph{domain},\footnote{Some papers use the term \emph{alphabet}.}
$\Pi$ is a set of constraints, and $w:\Pi\to\RR$ are positive weights for the
constraints. Each constraint $\pi\in\Pi$ is a pair $\tuple{(u,v),P}$, where
$(u,v)\in V^2$, called the constraint \emph{scope}, is a pair of distinct
variables from $V$, and $P:D^2\to\{0,1\}$ is a binary predicate. A CSP
instance is called \emph{Boolean} if $|D|=2$, i.e., if the domain is of size
two.\footnote{Some papers use the term \emph{binary} to mean domains of size
two. In this paper, \emph{Boolean} always refers to a domain of size two and
\emph{binary} always refers to the arity of the constraint(s).}

For a fixed binary predicate $P$, we denote by $\CSPP$ the class of CSP
instances in which all constraints use the predicate $P$. Note that if we take
$D=\{0,1\}$ and $P$ defined on $D^2$ by $P(x,y)=1$ iff $x\neq y$ then \CSP($P$)
corresponds to the cut problem. 

We say that a constraint $\pi=\tuple{(u,v), P}$ is \emph{satisfied} by an
assignment $A:V\to D$ if
$P(A(u),A(v)) = 1$. The value of an instance $I=(V,D,\Pi,w)$ under an assignment $A:V\to
D$ is
defined to be the total weight of satisfied constraints:
\begin{equation*}
\Val_{\I}(A) = \sum_{\pi=\tuple{(u,v),P} \in \Pi} w(\pi) P(A(u),A(v)).
\end{equation*}
For $0<\eps<1$, an $\eps$-sparsifier of $\I=(V,D,\Pi,w)$ is a re-weighted subinstance
$\I_{\eps}=(V,D,\Pi_\eps\subseteq\Pi,w_\eps)$ of $\I$ such that
\begin{equation*} 
\forall A:V\to D, \qquad\Val_{\I_{\eps}}(A) \in (1 \pm \eps)\Val_{\I}(A).
\end{equation*}
The goal is to 
obtain a sparsifier with the minimum number of constraints, i.e., $|\Pi_\eps|$.

A binary predicate $P$ is called \emph{sparsifiable} if for every instance
$\I\in\CSPP$ on $n=|V|$ variables
and for every $0<\eps<1$ there is an $\eps$-sparsifier for $I$ with
$O(\eps^{-2}n)$ constraints.

We call a (not necessarily Boolean or binary) predicate $P$ a \emph{singleton}
if $|P^{-1}(1)|=1$.

Filtser and Krauthgamer showed, among other results, the following
classification. Let $P$ be a binary Boolean predicate. Then, $P$ is
sparsifiable if and only if $P$ is not a singleton.\footnote{Filtser and
Krauthgamer use the term \emph{valued CSPs} for what we defined as CSPs. We
prefer CSPs in order to distinguish them from the much more general framework of
valued CSPs studied in~\cite{Kolmogorov17:sicomp}.} In other words, the only
predicates that are not sparsifiable are those with support of size one. 

\paragraph{Contributions}
As our main contribution, we identify in Theorem~\ref{thm:main} the precise
borderline of sparsifiability for binary predicates on arbitrary finite domains,
thus extending the work from~\cite{Fil} on Boolean predicates. Let $P$ be a
binary predicate defined on an arbitrary finite domain $D$. Then, $P$ is
sparsifiable if and only if $P$ does not ``contain'' a singleton subpredicate.
More precisely, we say that $P$ ``contains'' a singleton subpredicate if there
are two (not necessarily disjoint) subdomains $B, C\subseteq D$ with $|B|=|C|=2$
such that the restriction of $P$ onto $B\times C$ is a singleton predicate.

The crux of Theorem~\ref{thm:main} is the sparsifiability part, which is
established by a reduction to cut sparsifiers. Unlike in the classification of
binary Boolean predicates from~\cite{Fil}, we do not rely on a case analysis
that differs for different sparsifiable predicates but instead give a simpler
argument for all sparsifiable predicates. The idea is to reduce (the graph of)
any CSP instance, as was done in~\cite{Fil}, to a new graph via the so-called bipartite double
cover~\cite{Brualdi80:jgt}. However, there is no natural assignment in the
new graph (as it was in the Boolean case in~\cite{Fil}). In order to
overcome this, we define a graph $G^P$ whose edges correspond to the support of
the predicate $P$. Using a simple combinatorial argument, we show (in
Proposition~\ref{prop:bipartite}) that, under the assumption that $P$ does not
``contain'' a singleton subpredicate, the bipartite complement of $G^P$ is a
collection of bipartite cliques. This special structure allows us to find a good
assignment in the new graph.

In view of Filtser and Krauthgamer's work~\cite{Fil}, one might conjecture that
$P$ is sparsifiable if and only if $P$ is not a singleton. While it is easy to
show that if a (possibly non-binary and non-Boolean) predicate $P$ is a
singleton then $P$ is not sparsifiable (cf. Section~\ref{sec:sufficient} in the
appendix), 
our results show that the borderline of sparsifiability lies elsewhere. In
particular, by Theorem~\ref{thm:main}, there are binary non-Boolean predicates
that are not sparsifiable but are not singletons. Also, there are non-binary
Boolean predicates that are not sparsifiable but are not singletons
(cf. Section~\ref{sec:sufficient}).

We remark that the term ``sparsification'' is also used in an unrelated line of
work in which the goal is, given a CSP instance, to reduce the number
of constraints without changing satisfiability of the instance; see, e.g.,~\cite{Chen18:ipec}.

\section{Classification of Binary Predicates}

Throughout the paper we denote by $n=|V|$ the number of variables of a given CSP
instance.

The following classification of binary Boolean predicates is from~\cite{Fil}.

\begin{theorem}[\protect{\cite[Theorem~3.7]{Fil}}]\label{thm:Boolbinary}
  Let $P:\{0,1\}^2\to\{0,1\}$ be a binary Boolean predicate. Let $0<\eps<1$.
  \begin{enumerate}
      \item If $P$ is a singleton then there exists an instance $I$ of $\CSPP$ such that every
        $\eps$-sparsifier of $I$ has $\Omega(n^2)$ constraints.
      \item Otherwise, for every instance $I$ of $\CSPP$ there exists an
        $\eps$-sparsifier of $I$ with $O(\eps^{-2}n)$ constraints.
  \end{enumerate}
\end{theorem}

We denote by $\binom{D}{2}=\{B\subseteq D: |B|=2\}$ the set of two-element subsets
of $D$. For a binary predicate $P:D^2\to\{0,1\}$ and $B,C \in \binom{D}{2}$,
$\Crestrict{P}{B\times C}$ denotes the restriction of $P$ onto $B\times C$.

The following is our main result, generalising Theorem~\ref{thm:Boolbinary} to
arbitrary finite domains.

\begin{theorem}[\textbf{Main}]\label{thm:main}
  Let $P:D^2\to\{0,1\}$ be a binary predicate, where $D$ is a finite set with $|D|\geq 2$. Let $0<\eps<1$.
  \begin{enumerate}
      \item If there exist $B,C \in \binom{D}{2}$ such that
        $\Crestrict{P}{B\times C}$ is a singleton then there exists an instance
        $I$ of $\CSPP$ such that every $\eps$-sparsifier of $I$ has $\Omega(n^2)$ constraints.
      \item Otherwise, for every instance $I$ of $\CSPP$ there exists an
        $\eps$-sparsifier of $I$ with $O(\eps^{-2}n)$ constraints.
  \end{enumerate}
\end{theorem}

The rest of this section is devoted to proving Theorem~\ref{thm:main}.

First we introduce some useful notation. We set $[r]=\{0,1,\ldots,r-1\}$. We
denote by $X\sqcup Y$ the disjoint union of $X$ and $Y$. For any $r\geq 2$, we
define $r\textnormal{-}\Cut:[r]^2\to\{0,1\}$ by $r\textnormal{-}\Cut(x,y)=1$ if
and only if $x\neq y$. 

Given an instance $\I = (V,D,\Pi,w) \in \CSPP$, we denote by $G^{\I}$ the
\emph{corresponding graph} of $\I$; that is, $G^{\I} = (V,E,w)$ is a weighted directed
graph with $E= \{(u,v):\tuple{(u,v),P} \in \Pi\}$ and $w(u,v) = w((u,v),P)$.
Conversely, given a weighted directed graph $G=(V,E,w)$ and a predicate $P:D^{2}
\to \{0,1\}$, the \emph{corresponding $\CSPP$ instance} is $\I^{G,P} =
(V,D,\Pi,w)$, where $\Pi=\{\tuple{e,P} : e \in E\}$ and $w(e,P) = w(e)$.
Hence, we can equivalently talk about instances of \CSP($P$) or (weighted
directed) graphs. Thus, an $\eps$-$P$-\emph{sparsifier} of a graph $G=(V,E,w)$
is a subgraph $G_{\eps} = (V, E_{\eps} \subseteq E, w_{\eps})$ whose
corresponding $\CSPP$ instance $\I^{G_{\eps},P}$ is an $\eps$-sparsifier of the
corresponding $\CSPP$ instance $\I^{G,P}$ of $G$.

Case~(1) of Theorem~\ref{thm:main} is established by the following result.

\begin{theorem}\label{thm:not}
  Let $P:D^2\to\{0,1\}$ be a binary predicate. Assume that there exist $B,C
  \in \binom{D}{2}$ such that $\Crestrict{P}{B\times C}$ is a singleton.
  For any $n$ there is a $\CSPP$ instance $I$ with $2n$ variables and $n^2$
  constraints such that for any
  $0<\eps<1$ it holds that any $\eps$-sparsifier of $I$ has $n^2$ constraints.
\end{theorem}

\begin{proof}
  Suppose $B=\{b,b'\}$, $C=\{c,c'\}$ and assume without loss of
  generality that $\Crestrict{P}{B\times C}^{-1}(1)=\{(b,c)\}$; that is, the
  support of $\Crestrict{P}{B\times C}$ is equal to $\{(b,c)\}$.
  Consider a $\CSPP$ instance $\I=(V,D,\Pi,w)$, where
\begin{itemize}
\item $V=X \sqcup Y$, $X=\{x_1, \ldots, x_n\}$, and $Y=\{y_1, \ldots, y_n\}$;
\item $\Pi=\{\pi_{ij}=\tuple{(x_i,y_j), P} : 1\leq i,j\leq n\}$;
\item $w$ are arbitrary positive weights.
\end{itemize}
We have $|\Pi|= n^2$. 
We note that $B$ and $C$ may not be disjoint.
We consider the family of assignments $A_{ij}:V \to B \cup C$ for $1\leq i,j\leq
  n$ such that $A_{ij}(x_i)=b$, $A_{ij}(x)=b'$ for every $x\in
  X\setminus\{x_i\}$, $A_{ij}(y_j)=c$, and $A_{ij}(y)=c'$ for every $y\in
  Y\setminus\{y_j\}$. Then, we have
\[P(A_{ij}(u,v))= \begin{cases}
  P(b,c) \ \ =1 & \quad \textnormal{if } u=x_i, v=y_j,\\
  P(b,c') \  =0 & \quad \textnormal{if } u=x_i, v \in Y \setminus \{y_j\},\\
  P(b',c) \ =0 & \quad \textnormal{if } u \in X \setminus \{x_i\}, v=y_{j},\\
  P(b',c') =0 & \quad \textnormal{if } u \in X \setminus \{x_i\}, v \in Y \setminus \{y_j\}.\\
\end{cases}\]
Therefore, \[\Val_{\I}(A_{ij})= \sum_{\pi \in \Pi}w(\pi)P(A_{ij}(\pi))=
  w(\pi_{ij}) > 0.\]
Hence, if $\I_{\eps}=(V,D,\Pi_{\eps},w_{\eps})$ is an $\eps$-sparsifier of $\I$,
  we must have that $\pi_{ij} \in \Pi_{\eps}$ for every $1\leq i,j\leq n$, as
  otherwise we would have \[\Val_{\I_{\eps}}(A_{ij})= \sum_{\pi \in \Pi_{\eps}}w_{\eps}(\pi)P(A_{ij}(\pi))= 0 \notin (1\pm \eps)\Val_{\I}(A_{ij}).\]
Therefore, we have $\Pi_{\eps}=\Pi$ and hence $|\Pi_{\eps}|=|\Pi|=n^2$.
\end{proof}

The main tool used in the proof of Theorem~\ref{thm:Boolbinary}~(2) from~\cite{Fil} is a graph transformation known as
the bipartite double cover~\cite{Brualdi80:jgt}, which allows for a reduction to
cut sparsifiers~\cite{Bats}. 

\begin{definition}
  
For a weighted directed graph $G=(V,E,w)$, the
  \emph{bipartite double cover} of $G$ is the weighted directed graph
$\gamma(G)=(V^{\gamma}, E^{\gamma}, w^{\gamma})$, where
\begin{itemize}
\item $V^{\gamma} = \{v^{(0)},v^{(1)}:v \in V\}$;
\item $E^{\gamma} = \{(u^{(0)}, v^{(1)}): (u,v) \in E\}$;
\item $w^{\gamma}(u^{(0)}, v^{(1)}) = w(u,v)$.
\end{itemize}
\end{definition}

Given an assignment $A:V \to [r]$, 
we let $\mathcal{A}=(A_{0}, \ldots, A_{r-1})$ be the induced $r$-partition of $V$, where $A_{j} = A^{-1}(j)$. For
a binary predicate $P:[r]^{2} \to \{0,1\}$ and an instance $\I=(V,[r],\Pi,w) \in
\CSP(P)$, we define $\Val_{\I}(\mathcal{A}) = \Val_{\I}(A)$. Moreover, for a
weighted directed graph $G$ and a binary predicate $P$, we define
$\Val_{G,P}(\mathcal{A})=\Val_{\I^{G,P}}(\mathcal{A})$. We denote the set of all
$r$-partitions of $V$ by $Part_{r}(V)$. 

For any $r$-partition $\mathcal{A}=(A_{0}, \ldots, A_{r-1})$ of the vertices
of $V$, let $A_{i}^{(j)} = \{v^{(j)}: v \in A_{i}\}$. Thus
$\mathcal{A}^{\gamma} = (A^{(0)}_{0},A^{(1)}_{0}, \ldots,
A^{(0)}_{r-1}, A^{(1)}_{r-1})$ is a $2r$-partition of the vertices of
$V^{\gamma}$.

We use an argument from the proof of Theorem~\ref{thm:Boolbinary}\,(2)
from~\cite{Fil} and apply it to non-Boolean predicates. 

\begin{proposition} \label{cutCoverBipartite}
Let $P:[r]^{2} \to \{0,1\}$ and $P':[r']^{2} \to \{0,1\}$ be binary predicates.
Suppose that there is a function $f_{P}: Part_{r}(V) \to Part_{r'}(V^{\gamma})$ such that for any weighted directed graph $G$ on $V$ and for any $r$-partition $\mathcal{A} \in Part_{r}(V)$ it holds that 
\[\Val_{G,P}(\mathcal{A})= \Val_{\gamma(G),P'} (f_{P}(\mathcal{A})),\]
where $\gamma(G)=(V^{\gamma},E^{\gamma},w^{\gamma})$ is the bipartite double cover of $G$. 
  If there is an $\eps$-$P'$-sparsifier of $\gamma(G)$ of size $g(n)$ then there
  is an $\eps$-$P$-sparsifier of $G$ of size $g(n)$.
\end{proposition}

\begin{proof}

  Given $G=(V,E,w)$, let $\gamma(G)_{\eps}=(V,E^{\gamma}_{\eps},w^{\gamma}_{\eps})$ be an
  $\eps$-$P'$-sparsifier (of size $g(n)$) of the bipartite double cover $\gamma(G)$ of $G$. Define
  a subgraph $G_{\eps}=(V,E_{\eps},w_{\eps})$ of $G$ 
  by $E_{\eps} = \{(u,v): (u^{(0)},v^{(1)}) \in E^{\gamma}_{\eps}\}$ and
  $w_{\eps}(u,v)=w^{\gamma}_{\eps}(u^{(0)}, v^{(1)})$. Note
  that $\gamma(G_{\eps}) = \gamma(G)_{\eps}$ and $E_{\eps} \subseteq E$.

Then, we have
  \begin{align*}
    \Val_{G_{\eps},P}(\mathcal{A}) & =
    \Val_{\gamma(G_{\eps}),P'}(f_{P}(\mathcal{A})) \\
    & = \Val_{\gamma(G)_{\eps},P'}
  (f_{P}(\mathcal{A}))
  \in (1 \pm \eps) \Val_{\gamma(G),P'} (f_{P}(\mathcal{A}))
  = (1 \pm \eps) \Val_{G,P}(\mathcal{A}).
\end{align*}
  Hence $G_{\eps}$ is also an $\eps$-$P$-sparsifier of $G$ of size $g(n)$. 
\end{proof}

We now focus on proving Case~(2) of Theorem~\ref{thm:main}. 
Assume that for any $B,C \in \binom{D}{2}$,
$\Crestrict{P}{B\times C}$ is not a singleton. Our strategy is to show that in
this case the value of a $\CSPP$ instance under any assignment can be expressed
as the value of a corresponding \CSP($\ell$-\Cut) instance (for some $\ell \leq
2|D|$) under the same assignment. 

For an undirected graph $G=(V,E)$ and a subset $U \subseteq V$, we denote the
vertex-induced subgraph on $U$ by $G[U]$ and its edge set by $E[U]$. For a
possibly disconnected undirected graph $G$, we denote the connected component
containing a vertex $v$ by $G_{v}=(V(G_{v}),E(G_{v}))$. Finally, we denote the
degree of vertex $v$ in graph $G$ by $\deg_{G}(v)$.

\begin{definition}
Let $G=(U\sqcup V,E)$ be an undirected bipartite graph. The \emph{bipartite complement} $\overline{G} = (U\sqcup V,\overline{E})$ of $G$ has the following edge set:
\[\overline{E}= \{\{u,v\}:u \in U, v \in V, \{u,v\} \notin E\}.\]
\end{definition}

The following property of bipartite graphs will be crucial in the proof of
Theorem~\ref{thm:spars}.

\begin{proposition} \label{prop:bipartite}
Let $G=(U\sqcup V,E)$ be a bipartite graph with $|U|=|V|=r$, $r\geq 2$. Assume that for any $u,u' \in U$ and $v,v' \in V$ we have $|E[\{u,u',v,v'\}]| \neq 1$. Then, for any $v \in U\sqcup V$ with
$\deg_{\overline{G}}(v) > 0$, $\overline{G}_{v}$ is a complete bipartite graph with partition classes $\{U \cap V(\overline{G}_{v})\}$ and $\{V \cap V(\overline{G}_{v})\}$.
\end{proposition}

\begin{proof}

For contradiction, assume that there are $u\in U$ and $v\in V$ such that
  $\{u,v\}\not\in\overline{E}$ but $u$ and $v$ belong to the same connected component
  of $\overline{G}$. Choose $u$ and $v$ with the shortest possible distance
  between them. Let $u=u_0,u_1,\ldots,u_k=v$ be a shortest path between
  $u$ and $v$ in $\overline{G}$, where $k\geq 3$ is odd. We will show that
  $|\overline{E}[\{u_0,u_1,u_{k-1},u_k\}]|=3$, which contradicts the assumption
  that $|E[\{u_0,u_1,u_{k-1},u_k\}]|\neq 1$.
 
  If $k=3$ then the claim holds since we assumed that
  $\{u_0,u_1\},\{u_1,u_2\},\{u_2,u_3\}\in\overline{E}$ and $\{u_0,u_3\}\not\in\overline{E}$.

  \begin{figure}[ht]
  \begin{center}
  \begin{tikzpicture}
    \draw node [circle,draw,thick] (u0) at (0,3) {};
    \node[left of=u0] {$u_0$};
    \draw node [circle,draw,thick] (u1) at (3,3) {};
    \node[right of=u1] {$u_1$};
    \draw node [circle,draw,thick] (u2) at (0,1.5) {};
    \node[left of=u2] {$u_2$};
    \draw node [circle,draw,thick] (u3) at (3,1.5) {};
    \node[right of=u3] {$u_3$};
    \draw node [circle,draw,thick] (u4) at (0,0) {};
    \node[left of=u4] {$u_4$};
    \draw node [circle,draw,thick] (u5) at (3,0) {};
    \node[right of=u5] {$u_5$};
    \draw [thick] (u0) -- (u1);
    \draw [thick] (u1) -- (u2);
    \draw [thick] (u2) -- (u3);
    \draw [thick] (u3) -- (u4);
    \draw [thick] (u4) -- (u5);
  \end{tikzpicture} 
  \end{center}
  \caption{Illustration of the proof of Proposition~\ref{prop:bipartite} for $k=5$.}\label{fig:k}
  \end{figure}
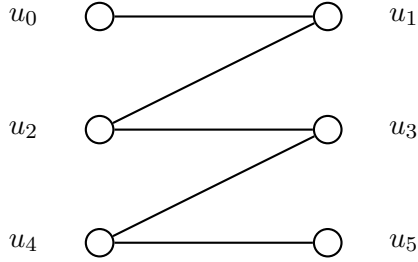
 
  Let $k\geq 5$.
  We will be done if we show that
  $\{u_1,u_{k-1}\}\in\overline{E}$, as by our assumptions
  $\{u_0,u_1\},\{u_{k-1},u_k\}\in\overline{E}$ and
  $\{u_0,u_k\}\not\in\overline{E}$. To this end, note that
  $\{u_0,u_{k-2}\}\in\overline{E}$ as otherwise $u_0$ and $u_{k-2}$ would
  be a pair of vertices with the required properties but of distance $k-2$,
  contradicting our choice of $u$ and $v$. Thus, $\{u_1,u_{k-1}\}\in\overline{E}$
  as otherwise we would have $|\overline{E}[\{u_0,u_1,u_{k-2},u_{k-1}\}]|=3$,
  which contradicts $|E[\{u_0,u_1,u_{k-2},u_{k-1}\}]| \neq 1$.
  (See Figure~\ref{fig:k} for an illustration of the case $k=5$.)
\end{proof}

Case~(2) of Theorem~\ref{thm:main} is established by the following result.

\begin{theorem} \label{thm:spars}
  Let $P:D^{2} \to \{0,1\}$ be a binary predicate such that for any $B,C \in
  \binom{D}{2}$ we have that $\Crestrict{P}{B\times C}$ is not a singleton.
  Then, for every $0<\eps<1$ and every instance $I$ of $\CSPP$ there is a
  sparsifier of $I$ with $O(\eps^{-2}n)$ constraints.
\end{theorem}

\begin{proof} 
  Let $I=(V,D,\Pi,w)$ be an instance of $\CSPP$ with $r=|D|$. Without loss of
  generality, we assume that $D=[r]$.
  Let $G=G^I=(V,E,w)$ be the corresponding (weighted directed) graph of $I$, and 
  let $\gamma(G) = (V^{\gamma},E^{\gamma},w^{\gamma})$ be the bipartite double
  cover of $G$. 
  Recall that for an assignment $A:V \to [r]$, we denote $A_{i}=A^{-1}(i)$. Thus,
  $\mathcal{A}=(A_{0}, \ldots, A_{r-1})$ forms an $r$-partition of $V$. 

  Our goal is to show the existence of a function $f_{P}: Part_{r}(V) \to
  Part_{\ell}(V^{\gamma})$ (for some fixed $\ell \leq 2r$) such that 
  \begin{equation}\label{eq:map}
    \forall A:V\to [r], \qquad\Val_{G,P}(\mathcal{A})= \Val_{\gamma(G),\textnormal{$\ell$-\Cut}} (f_{P}(\mathcal{A})).
  \end{equation}

  Assuming the existence of $f_{P}$, we can finish the proof as follows.
  Batson, Spielman, and Srivastava established the existence of a sparsifier of
  size $O(\eps^{-2}n)$ for any instance of $\CSP(2\textnormal{-}\Cut)$~\cite{Bats}.
  By~\cite[Section~6.2]{Fil}, this implies the existence of a sparsifier of size
  $O(\eps^{-2}n)$ for any instance of $\CSP(\ell\textnormal{-}\Cut)$.
  Consequently, by Proposition~\ref{cutCoverBipartite} and~(\ref{eq:map}), there is a
  sparsifier of size $O(\eps^{-2}n)$ for the instance $I^{G,P}=I$.

  In the proof of Theorem~\ref{thm:Boolbinary}\,(2) in~\cite{Fil}, 
  functions $f_P$ are given for any binary Boolean predicate $P$
  with support size $|P^{-1}(1)| \in \{0,2,4\}$. In what
  follows we give a construction of $f_P$ for an \emph{arbitrary} binary
  predicate $P:[r]^{2} \to \{0,1\}$ with $r \geq 2$ from the statement of the
  theorem.

  Although the bipartite double cover is commonly defined as a directed graph, in this proof we will consider the \emph{undirected} bipartite double cover
  $\gamma(G)$ of $G$.\footnote{We had defined the bipartite double cover as a
  directed graph. However, here it is easier to deal with undirected graphs, as
  since $\ell$-$\Cut$ is a symmetric predicate, the direction of the edges makes no
  difference. Furthermore, notice that by the way the bipartite double cover is
  constructed, removing the direction does \emph{not} turn the graph into a
  multigraph.} We also define an auxiliary graph $G^{P} =(V^{P},E^{P})$, where 
  \[V^{P} = \{v_{0}, v'_{0}, \ldots,
  v_{r-1},v'_{r-1}\},\] \[E^{P} = \{\{v_{i},v'_{j}\}: P(i,j)=1\}.\] 

  Let $\ell$ be the number of connected components of $\overline{G^{P}}$, the
  bipartite complement of $G^{P}$. By
  definition, $\ell \leq |V^{P}|=2r$. 
  
  The desired function $f_P$ satisfying~(\ref{eq:map}) corresponds to a map $c:V^{P} \to [\ell]$ on the vertices of $G^{P}$ with the following property:

  \[(*)\quad \forall i,j \in [r] \quad \begin{cases}
\{v_{i},v'_{j}\} \in E^{P} \implies c(v_{i}) \neq c(v'_{j})\\
\{v_{i},v'_{j}\} \notin E^{P} \implies c(v_{i}) = c(v'_{j}).\\
\end{cases}\]

We call such maps \emph{colourings}. Indeed, the colouring $c$ induces, for
  $\mathcal{A}$, an
  assignment $A^{\gamma}:V^{\gamma} \to [\ell]$ of the vertices of $\gamma(G)$ which satisfies 
  \[A^{\gamma}(u) = c(v_{A(u)}) \quad \mbox{and} \quad A^{\gamma}(u') = c(v'_{A(u)})\]
and which, in turn, induces a partition $\{U_{i}\}_{i=0}^{\ell-1}$ of
  $V^{\gamma}$ with $U_{i} = (A^{\gamma})^{-1}(i)$. We define
  $f_{P}(\mathcal{A}) = (U_{0}, \ldots, U_{\ell-1})$. Now for any $u,v \in V$ and for any assignment $A:V \to [r]$, we have
\begin{align*}
P(A(u),A(v))=1 & \iff \{v_{A(u)},v'_{A(v)}\} \in E^{P} \\
& \iff c(v_{A(u)}) \neq c(v'_{A(v)}) \\
& \iff A^{\gamma}(u) \neq A^{\gamma}(v') \\
& \iff \textnormal{$\ell$-\Cut}(A^{\gamma}(u), A^{\gamma}(v')) = 1.
\end{align*}

Moreover, by the definition of the bipartite double cover, we have $w(u,v)=w^{\gamma}(u,v')$ for all $u,v \in V$, implying that 
\begin{align*}
  \Val_{G,P}(\mathcal{A}) &= \Val_{G,P}(A_{0}, \ldots, A_{r-1})  = \sum_{(u,v) \in E}w(u,v)P(A(u),A(v)) \\
& = \sum_{(u,v') \in E^{\gamma}} w^{\gamma}(u,v')
  \textnormal{$\ell$-\Cut}(A^{\gamma}(u), A^{\gamma}(v')) = \Val_{\gamma(G),\textnormal{$\ell$-\Cut}}(A^{\gamma}) \\
  & = \Val_{\gamma(G),\textnormal{$\ell$-\Cut}}(U_{0}, \ldots, U_{\ell-1}) =
  \Val_{\gamma(G),\textnormal{$\ell$-\Cut}}(f_{P}(\mathcal{A}))
\end{align*}
as required.

While a colouring does not exist for an arbitrary bipartite graph, we now argue
that a colouring does exist if the auxiliary graph $G^{P}$ arises from a
predicate $P$ from the statement of the theorem. Since for any $B,C \in
\binom{[r]}{2}$ we have $|\Crestrict{P}{B\times C}^{-1}(1)| \neq 1$, $G^{P}$
satisfies the assumptions of Proposition~\ref{prop:bipartite}. Therefore, the
$\ell$ separate connected components which form its bipartite complement
$\overline{G^{P}}$ are complete bipartite graphs. We can assign one of the
$\ell$ colours to each connected component to get a colouring for the graph
$G^{P}$. We now show that this colouring satisfies $(*)$. 
(See Figure~\ref{fig:col} for an example of $G^{P}$, $\overline{G^{P}}$, and the
colouring with $\ell=3$ satisfying $(*)$ for a particular predicate $P$ on a four-element
domain.)
\bigskip

  \begin{figure}[ht]
  \begin{center}
  \begin{tikzpicture}
    \draw node [circle,draw,thick,fill=blue] (a0) at (0,4.5) {};
    \draw node [circle,draw,thick,fill=green] (b0) at (3,4.5) {};
    \draw node [circle,draw,thick,fill=red] (a1) at (0,3) {};
    \draw node [circle,draw,thick,fill=blue] (b1) at (3,3) {};
    \draw node [circle,draw,thick,fill=red] (a2) at (0,1.5) {};
    \draw node [circle,draw,thick,fill=red] (b2) at (3,1.5) {};
    \draw node [circle,draw,thick,fill=green] (a3) at (0,0) {};
    \draw node [circle,draw,thick,fill=red] (b3) at (3,0) {};
    \draw [thick] (a0) -- (b0);
    \draw [thick] (a0) -- (b2);
    \draw [thick] (a0) -- (b3);
    \draw [thick] (a1) -- (b0);
    \draw [thick] (a1) -- (b1);
    \draw [thick] (a2) -- (b0);
    \draw [thick] (a2) -- (b1);
    \draw [thick] (a3) -- (b1);
    \draw [thick] (a3) -- (b2);
    \draw [thick] (a3) -- (b3);
    \draw node (l) at (1.5,-1) {$G^{P}$};

    \begin{scope}[shift={(+7,0)}] 
    \draw node [circle,draw,thick,fill=blue] (a0) at (0,4.5) {};
    \draw node [circle,draw,thick,fill=green] (b0) at (3,4.5) {};
    \draw node [circle,draw,thick,fill=red] (a1) at (0,3) {};
    \draw node [circle,draw,thick,fill=blue] (b1) at (3,3) {};
    \draw node [circle,draw,thick,fill=red] (a2) at (0,1.5) {};
    \draw node [circle,draw,thick,fill=red] (b2) at (3,1.5) {};
    \draw node [circle,draw,thick,fill=green] (a3) at (0,0) {};
    \draw node [circle,draw,thick,fill=red] (b3) at (3,0) {};
    \draw [thick] (a0) -- (b1);
    \draw [thick] (a1) -- (b2);
    \draw [thick] (a1) -- (b3);
    \draw [thick] (a2) -- (b2);
    \draw [thick] (a2) -- (b3);
    \draw [thick] (a3) -- (b0);
      \draw node (l) at (1.5,-1) {$\overline{G^{P}}$};
    \end{scope}

  \end{tikzpicture} 
  \end{center}
    \caption{An example of $G^{P}$ and $\overline{G^{P}}$ from the proof of
    Theorem~\ref{thm:spars}. The (vertex) colouring indicates the bicliques of 
    $\overline{G^{P}}$.}\label{fig:col}
  \end{figure}
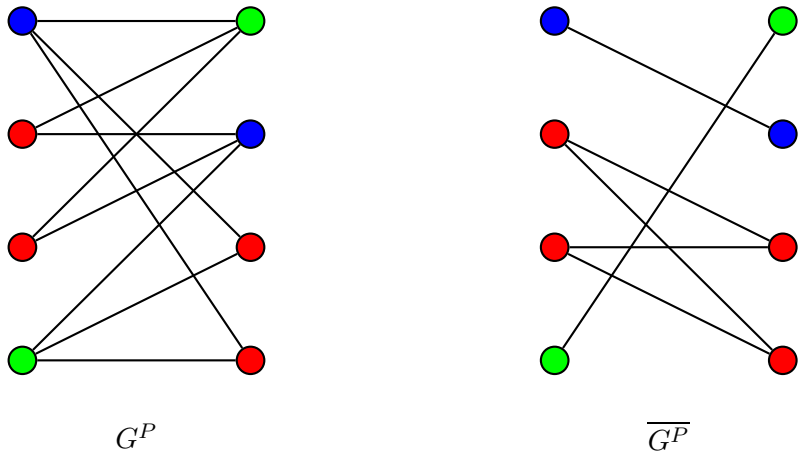

Indeed, if
$\{v_i,v'_j\}\in E^P$ then $\{v_i,v'_j\}$ is not an edge in $\overline{G^{P}}$.
Hence $v_i$ and $v'_j$ are in different connected components of
$\overline{G^{P}}$ and thus $v_i$ and $v'_j$ are assigned different colours.
Similarly, if $\{v_i,v'_j\}\not\in E^P$ then $\{v_i,v'_j\}$ is an edge in
$\overline{G^{P}}$. Hence $v_i$ and $v'_j$ are in the same connected component 
of $\overline{G^{P}}$ and thus are assigned the same colour.
\end{proof}

\section{Conclusion}

For simplicity, we have only presented our main result on binary CSPs over a
single domain. However, it is not difficult to extend our result to the
so-called \emph{multisorted} binary CSPs, in which different variables come with
possibly different domains.  We discuss this in the appendix.

We have classified \emph{binary} CSPs (on finite domains) but much more work
seems required for a full classification of non-binary CSPs. We have made some
initial steps. 

For any $k\geq 3$, the $k$-ary Boolean ``not-all-equal'' predicate
$k\textnormal{-}\NAE:\{0,1\}^{k} \to \{0,1\}$ is defined by
$k\textnormal{-}\NAE^{-1}(0)=\{(0, \ldots,0),(1, \ldots,1)\}$. Kogan and
Krauthgamer showed that the $k\textnormal{-}\NAE$ predicates, which correspond
to hypergraph cuts, are sparsifiable~\cite[Theorem~3.1]{Kog}. By extending
bipartite double covers for graphs in a natural way to $k$-partite $k$-fold
covers (in Section~\ref{sec:hypergraph}) we obtain sparsifiability for the class
of $k$-ary predicates that can be rewritten in terms of $k\textnormal{-}\NAE$.
On the other hand, we identify (in Section~\ref{sec:sufficient}) a whole class
of predicates that are not sparsifiable, namely those $k$-ary predicates that
contain a singleton $\ell$-cube for some $\ell\leq k$. However, most predicates
do not fall in either of these two categories; that is, predicates that cannot
be proved sparsifiable via $k$-partite $k$-fold covers but also cannot be proved
non-sparsifiable via the current techniques. An example of such predicates are
the ``parity'' predicates (cf. Section~\ref{sec:parity} of the appendix).

\section*{Acknowledgements} 

The authors thank all reviewers of the extended abstract~\cite{bz19:stacs} and
this full version of the paper for useful comments.


\newcommand{\noopsort}[1]{}\newcommand{\Zivny}{\noopsort{ZZ}\v{Z}ivn\'y}

\appendix

\section{Extensions}

\subsection{Constraint Satisfaction Problems} \label{sec:csp}

There are several natural and well-studied extensions of the binary CSP
framework: (i) \emph{non-binary} CSPs, in which constraints are of arity larger
than two; (ii) \emph{multisorted} CSPs, in which different variables have
possibly different domains; and (iii) CSPs with different types of constraint
predicates, leading to \emph{constraint languages}. 

\begin{definition}
An instance of the \emph{constraint satisfaction problem} (CSP) is a quadruple
$I = (V,\mathcal{D},\Pi,w)$ where $V=\{v_1,\ldots,v_n\}$ is a set of
variables, $\mathcal{D}=\{D(v_1), \ldots, D(v_n)\}$ is a set of domains, one
for each variable, $\Pi$ is a set of constraints, and $w:\Pi \to \RR$ are
positive weights for the constraints. Each constraint $\pi \in \Pi $ is a pair
$\tuple{\mathbf{v},P}$, where $\textbf{v}=\tuple{v_{i_1},\ldots,v_{i_k}}\in
V^k$ is an ordered $k$-tuple of distinct variables from $V$ and
$P:D(v_{i_1})\times\ldots\times D(v_{i_k})\to\{0,1\}$ is a $k$-ary predicate
on the Cartesian product of the corresponding domains.
\end{definition}

The elements from $\cup_{D\in\mathcal{D}} D$ are called \emph{labels}.

For a fixed predicate $P$, we denote by $\CSPP$ the class of CSP
instances in which all constraints use the predicate $P$. 

We say that an assignment $A: V \to \cup_{D \in \mathcal{D}} D$ is \emph{valid}
if each variable $v \in V$ is assigned a label that belongs to the
intersection of the domains of all the constraint predicates whose scope
contains $v$. For a vector $\textbf{v} \in V^{k}$ and an assignment $A:V \to
\cup_{D \in \mathcal{D}} D$, we denote by $A(\textbf{v})$ the entry-wise
application of $A$ to $\textbf{v}$. Given a predicate $P:
D(v_1)\times\ldots\times D(v_k) \to \{0,1\}$, we say that a constraint $\pi =
\tuple{\textbf{v}, P}$ is \emph{satisfied} by an assignment $A$ if
$P(A(\textbf{v})) = 1$. 

The \emph{value} of an instance $\I=(V,\mathcal{D},\Pi,w)$ under
assignment $A: V \to \cup_{D \in \mathcal{D}} D$ is given by the total weight
of the constraints satisfied by $A$:

\begin{equation}
  \Val_{\I}(A) = \sum_{\pi=\tuple{\textbf{v},P} \in \Pi} w(\pi) P(A(\textbf{v})).
\end{equation}

For $0<\eps<1$, an $\eps$-\emph{sparsifier} of $\I=(V,\mathcal{D},\Pi,w)$
  is a re-weighted subinstance $\I_{\eps} = (V, \mathcal{D}, \Pi_{\eps}
  \subseteq \Pi, w_{\eps})$ of $I$ such that
for all valid assignments $A$ of the variables in $V$, 
\begin{equation} \label{value2}
\Val_{\I_{\eps}}(A) \in (1 \pm \eps)\Val_{\I}(A).
\end{equation}

Given an instance $\I = (V,\mathcal{D},\Pi,w) \in \CSPP$ for a $k$-ary $P$, we
will call the \emph{corresponding hypergraph} of $\I$ the weighted directed
$k$-uniform hypergraph $H^{\I} = (V,E,w)$, where $E= \{\textbf{v}:
\tuple{\textbf{v},P} \in \Pi\}$ and $w(\textbf{v}) = w(\textbf{v},P)$. 
Conversely, given a weighted directed $k$-uniform hypergraph $H=(V,E,w)$ and a predicate $P:D^{k} \to
\{0,1\}$, the \emph{corresponding $\CSPP$ instance} is $\I^{H,P} =
(V,\mathcal{D},\Pi,w)$, where $\mathcal{D}=\{D\}$, $\Pi=\{\tuple{e,P} : e \in
E\}$, and $w(e,P) = w(e)$.
Hence, we can equivalently talk about instances of \CSP($P$) or hypergraphs. 
Thus, an $\eps$-$P$-\emph{sparsifier} of a hypergraph $H=(V,E,w)$ is a partial
subhypergraph\footnote{A partial subhypergraph is obtained by removing
hyperedges while keeping the vertex set unchanged.} $H_{\eps} = (V, E_{\eps} \subseteq E, w_{\eps})$ whose corresponding $\CSPP$ instance $\I^{H_{\eps},P}$ is an $\eps$-sparsifier of the corresponding $\CSPP$ instance $\I^{H,P}$ of $H$.

\subsection{Multisorted Binary Predicates} \label{sec:domains}

The following result is a multisorted extension of Theorem~\ref{thm:main}.

\begin{theorem}\label{thm:maingen}
  Let $P:D\times E\to\{0,1\}$ be a binary predicate, where $D$ and $E$ are
  finite sets with $|D|,|E|\geq 2$. Let $0<\eps<1$.
  \begin{enumerate}
    \item If there exist $B\in\binom{D}{2}$ and $C\in\binom{E}{2}$ such that
        $\Crestrict{P}{B\times C}$ is a singleton then there exists an instance
        $I$ of $\CSPP$ such that every $\eps$-sparsifier of $I$ has $\Omega(n^2)$ constraints.
      \item Otherwise, for every instance $I$ of $\CSPP$ there exists an
        $\eps$-sparsifier of $I$ with $O(\eps^{-2}n)$ constraints.
  \end{enumerate}
\end{theorem}

An inspection of the proof of Theorem~\ref{thm:not} reveals that the proof
establishes Case~(1) of Theorem~\ref{thm:maingen}.
The proof of Theorem~\ref{thm:maingen}\,(2) is essentially identical to the proof of
Theorem~\ref{thm:spars}. The main difference is that, using the notation
from the proof of Theorem~\ref{thm:spars}, the bipartite double cover
$\gamma(G)=(V^\gamma,E^\gamma,w^\gamma)$ of $G$ may contain vertices of degree
zero. Let $Z=\{v^\gamma: \deg_{\gamma(G)}(v)=0\}$ be all such vertices. Let $\tau(G)=(V^{\tau},E^{\tau},w^{\tau})$
be the subgraph of $\gamma(G)$ induced by $V^\gamma\setminus Z$. Then,
for any valid assignment $A:V^\gamma\to D\cup E$ we have
  \begin{equation}
    \Val_{\tau(G),P}(A')= \Val_{\gamma(G),P}(A),
  \end{equation}
where $A'$ is the restriction of $A$ to $V^\tau$. Working with $\tau(G)$ instead
of $\gamma(G)$, the rest of the proof proceeds identically to the proof of
Theorem~\ref{thm:spars}, except for applying Proposition~\ref{prop:bipartite} 
to bipartite graphs whose left part is $D$ and the right part is $E$.
This last step is fine since the proof of Proposition~\ref{prop:bipartite} is
not affected if the given bipartite graph has parts of different sizes.

\begin{remark}\label{rem:langs}

For a fixed \emph{set} $\Gamma$ of predicates, we denote by $\CSP(\Gamma)$ the
class of CSP instances in which all constraints use predicates from $\Gamma$.
$\Gamma$ is often called a constraint language.  Filtser and Krauthgamer
considered sparsifiability of binary Boolean CSPs of the from $\CSP(\Gamma)$,
i.e., CSPs with multiple binary Boolean predicates~\cite[Section~5]{Fil}.
Under the assumption that no two constraints act on the same list of
variables, any instance $\I$ of $\CSP(\Gamma)$ can be partitioned 
into disjoint CSP instances according to the predicates in the constraints. By
finding a sparsifier for each of these instances, the union of the sparsifiers
  yields a sparsifier for $\I$. Thus our main sparsifiability result
  (Case~(2) of Theorem~\ref{thm:main} and its multisorted generalisation,
  Case~(2) of Theorem~\ref{thm:maingen}) trivially
extends to $\CSP(\Gamma)$ for any $\Gamma$ that consists of predicates that do
not contain singleton subpredicates.

\end{remark}

\subsection{Hypergraph Covers} \label{sec:hypergraph}

We generalise the notion of the bipartite double cover for graphs
from~\cite{Brualdi80:jgt} in a natural way to that of a $k$-partite $k$-fold
cover for hypergraphs, as this will be useful in the proof of
Theorem~\ref{thm:singleton}. The case of $k=2$ in the following definition
corresponds to the bipartite double cover.

\begin{definition}
  
For a weighted directed $k$-uniform hypergraph $H=(V,E,w)$, the
\emph{$k$-partite $k$-fold cover} of $H$ is the weighted directed $k$-uniform hypergraph
$\gamma(H)=(V^{\gamma}, E^{\gamma}, w^{\gamma})$, where
\begin{itemize}
\item $V^{\gamma} = \{v^{(0)},v^{(1)}, \ldots, v^{(k-1)}:v \in V\}$;
\item $E^{\gamma} = \{(v_{1}^{(0)}, \ldots, v_{k}^{(k-1)}): (v_{1}, \ldots, v_{k}) \in E\}$;
\item $w^{\gamma}((v_{1}^{(0)}, \ldots, v_{k}^{(k-1)})) = w(v_{1}, \ldots, v_{k})$.
\end{itemize}
\end{definition}

Given an assignment $A:V \to [r]$, we let $\mathcal{A}=(A_{0}, \ldots, A_{r-1})$
be the induced $r$-partition of $V$, where $A_{j} = A^{-1}(j)$. For a predicate
$P:[r]^{k} \to \{0,1\}$ and an instance $\I=(V,[r],\Pi,w) \in \CSP(P)$, we
define $\Val_{\I}(\mathcal{A}) = \Val_{\I}(A)$. Moreover, for a weighted
directed $k$-uniform hypergraph $H$ and a $k$-ary predicate $P$, we define
$\Val_{H,P}(\mathcal{A})=\Val_{\I^{H,P}}(\mathcal{A})$. We denote the set of all
$r$-partitions of $V$ by $Part_{r}(V)$. 

For any $r$-partition $\mathcal{A}=(A_{0}, \ldots, A_{r-1})$ of the vertices of
$V$, let $A_{i}^{(j)} = \{v^{(j)}: v \in A_{i}\}$. Thus $\mathcal{A}^{\gamma} =
(A^{(0)}_{0},\ldots,A^{(k-1)}_{0}, \ldots, A^{(0)}_{r-1}, \ldots,
A^{(k-1)}_{r-1})$ is a $kr$-partition of the vertices of $V^{\gamma}$.

We use an argument from the proof Theorem~\ref{thm:Boolbinary}\,(2)
from~\cite{Fil} and apply it to non-binary, non-Boolean predicates. 

\begin{proposition} \label{cutCover}
Let $P:[r]^{k} \to \{0,1\}$ and $P':[r']^{k} \to \{0,1\}$ be $k$-ary predicates.
Suppose that there is a function $f_{P}: Part_{r}(V) \to Part_{r'}(V^{\gamma})$ such that for any weighted directed $k$-uniform hypergraph $H$ on $V$ and for any $r$-partition $\mathcal{A} \in Part_{r}(V)$ it holds that 
\[\Val_{H,P}(\mathcal{A})= \Val_{\gamma(H),P'} (f_{P}(\mathcal{A})),\]
where $\gamma(H)=(V^{\gamma},E^{\gamma},w^{\gamma})$ is the $k$-partite $k$-fold cover of $H$. 
  If there is an $\eps$-$P'$-sparsifier of $\gamma(H)$ of size $g(n)$ then there is
  an $\eps$-$P$-sparsifier of $H$ size $g(n)$.
\end{proposition}
\begin{proof}

Given $H=(V,E,w)$, let $\gamma(H)_{\eps}=(V,E^{\gamma}_{\eps},w^{\gamma}_{\eps})$ be an
  $\eps$-$P'$-sparsifier of the $k$-partite $k$-fold cover $\gamma(H)$. Define
  a partial subhypergraph $H_{\eps}=(V,E_{\eps},w_{\eps})$ of $H$ 
  by $E_{\eps} = \{(v_{1}, \ldots,v_{k}): (v_{1}^{(0)}, \ldots,v_{k}^{(k-1)})
  \in E^{\gamma}_{\eps}\}$ and $w_{\eps}((v_{1},
  \ldots,v_{k}))=w^{\gamma}_{\eps}(v_{1}^{(0)}, \ldots,v_{k}^{(k-1)})$. Note
  that $\gamma(H_{\eps}) = \gamma(H)_{\eps}$ and $E_{\eps} \subseteq E$.

Then, we have
  \begin{align*}
    \Val_{H_{\eps},P}(\mathcal{A}) & =
    \Val_{\gamma(H_{\eps}),P'}(f_{P}(\mathcal{A})) \\
    & = \Val_{\gamma(H)_{\eps},P'}
  (f_{P}(\mathcal{A}))
  \in (1 \pm \eps) \Val_{\gamma(H),P'} (f_{P}(\mathcal{A}))
  = (1 \pm \eps) \Val_{H,P}(\mathcal{A}).
\end{align*}
Hence $H_{\eps}$ is also an $\eps$-$P$-sparsifier of $H$ of size $g(n)$.
\end{proof}

\subsection{Non-Sparsifiability and Singleton Predicates} \label{sec:sufficient}

We identify two simple sufficient conditions for a predicate \emph{not} to be
sparsifiable, namely the \emph{singleton $\ell$-cube}
(Proposition~\ref{prop:singletoncube}) and the \emph{unused label}
(Proposition~\ref{prop:unused}). We then use these conditions to show that
singleton predicates are not sparsifiable.

The idea of a singleton $\ell$-cube is essentially an $\ell$-ary singleton
subpredicate with Boolean domain.

\begin{definition} \label{lcube}

A $k$-ary predicate $P:D^{k} \to \{0,1\}$ \emph{contains a singleton $\ell$-cube} for some
$2 \leq \ell \leq k$ if there exist subdomains
$\{D_{j}=\{d^{j}_{0},d^{j}_{1}\}\}_{j=1}^{\ell} \in \binom{D}{2}$, indices
  $\{n_j\}_{j=1}^{\ell} \in \{0,1\}$, and a permutation $\sigma$ on
  $\{1,2,\ldots,k\}$ 
  such that there
exist $x_{\ell + 1}, \ldots, x_{k} \in D$ which satisfy \[P(\sigma(d^{1}_{n_1}, \ldots,
d^{\ell}_{n_\ell}, x_{\ell +1}, \ldots, x_{k}))=1\] and for all $y_{\ell + 1}, \ldots,
y_{k} \in D$, for all $i_{j} \in \{0,1\}$, \[P(\sigma(d^{1}_{i_{1}}, \ldots,
d^{\ell}_{i_{\ell}}, y_{\ell +1}, \ldots, y_{k}))=1 \quad \implies \quad i_{j}=n_j
\textnormal{ for all }j = 1, \ldots, \ell.\] 

\end{definition}

\begin{proposition}[\textbf{Singleton $\ell$-cube}]\label{prop:singletoncube} 

Let $P:D^{k} \to \{0,1\}$ be a $k$-ary predicate which contains a singleton
$\ell$-cube. Then, there exists a weighted directed $k$-uniform hypergraph
$H=(V,E,w)$ with $|V|=n$ such that for every $0<\eps<1$ and for every partial
subhypergraph $H_{\eps}=(V,E_{\eps},w_{\eps})$ of $H$ which
satisfies~(\ref{value2}), we have $|E_{\eps}| = \Omega(n^{\ell})$.

\end{proposition}

\begin{proof}
Let $\{D_{j}=\{d^{j}_{0},d^{j}_{1}\}\}_{j=1}^{\ell}$ and $\{n_j\}_{j=1}^{\ell}$
be as in Definition~\ref{lcube}. Without loss of generality, assume that $\sigma$ is
the identity permutation. 

Let $H=(V,E,w)$ be a weighted directed $k$-uniform hypergraph on $n = kq$
vertices with $V= V_{1} \sqcup \ldots \sqcup V_{k}$, $|V_{i}|=q$
for $i=1, \ldots, k$, and $E=\{(u_{1}, \ldots, u_{k}): u_{i}
  \in V_{i}\}$. Notice that $|E| = q^k$. Take an arbitrary hyperedge
  $f=(v_{1}, \ldots, v_{k}) \in E$. By construction, $v_{j} \in V_{j}$ for all
  $j$. Furthermore, pick some $x_{\ell+1}, \ldots, x_{k}$ such that
  $P(d^{1}_{n_1}, \ldots, d^{\ell}_{n_\ell},x_{\ell+1}, \ldots, x_{k})=1$.

Define the assignment
\[A^{f}:V \to D, \quad \quad \begin{cases}
 A^{f}(v_{j})=d^{j}_{n_j} & \quad \textnormal{ for }j \leq \ell,\\
 A^{f}(v) = d^{j}_{1-n_j} \quad \forall v \in V_{j} \setminus \{v_{j}\} & \quad \textnormal{ for }j \leq \ell,\\
 A^{f}(v) = x_{j} \quad  \forall v \in V_{j} & \quad \textnormal{ for }\ell +1 \leq j \leq k.\\
\end{cases}\]

Notice that $P(A^{f}(u_{1}, \ldots, u_{k})) = 1 \iff u_{j}=v_{j}$ for all $j
\leq \ell$. Therefore, at least one of the $q^{k-\ell}$ edges whose first
$\ell$ variables are $v_{1}, \ldots, v_{\ell}$ must belong to $E_{\eps}$
for~(\ref{value2}) to be satisfied. We repeat the same argument for all
$q^{\ell}$ combinations of vertices $(v_{1}, \ldots, v_{\ell}) \in V_{1}
  \times \ldots \times V_{\ell}$. Thus $|E_{\eps}| \geq q^{\ell} =
  \Theta(n^{\ell})$, as $k$ is a constant, and $|E_{\eps}|= \Omega(n^{\ell})$ as required.
\end{proof}

\begin{example}\label{ex:non}
Let $P:\{0,1\}^3 \to \{0,1\}$ be the ternary Boolean predicate defined by
$P^{-1}(1)=\{(0,0,0), (0,0,1)\}$. Note that $P$ is not a singleton. $P$
contains a singleton $2$-cube (e.g., on the first two coordinates) and thus it is not
sparsifiable by Proposition~\ref{prop:singletoncube}.
\end{example}

Our second sufficient condition for not being sparsifiable is the idea of an
unused label. An unused label is an element of the domain which never appears in
the tuples that belong to the predicate's support set.

\begin{proposition}[\textbf{Unused label}] \label{prop:unused}

Let $P:D^{k} \to \{0,1\}$ be a $k$-ary predicate with $P^{-1}(1) \neq \emptyset$.
Suppose that there exists $z \in D$ such that, for all $x_{1}, \ldots, x_{k-1}
\in D$ and for all permutations $\sigma$ on $\{1,2,\ldots,k\}$, $P(\sigma(x_{1}, \ldots,
x_{k-1},z))=0$. Then, for every weighted directed $k$-uniform hypergraph
$H=(V,E,w)$, for every $0<\eps<1$, and for every partial subhypergraph
$H_\eps=(V,E_\eps\subseteq E,w_\eps)$ of $H$ which satisfies~(\ref{value2}),
we have $|E_\eps|=\Omega(|E|)$.

\end{proposition}

\begin{proof}
  Let $H=(V,H,w)$, $0<\eps<1$, and $H_\eps=(V,E_\eps\subseteq E, w_\eps)$ be as in the statement.
  We will show that $|E_\eps|=\Omega(|E|)$.
  
  Consider some tuple $(a_{1}, \ldots, a_{k}) \in P^{-1}(1)$. By assumption, $z$ does not appear in any tuple which belongs to $P^{-1}(1)$ and therefore we must have $a_{j} \neq z$ for all $j$. Pick
  a hyperedge $f = (u_{1}, \ldots, u_{k}) \in E$ and let 
  $U=\{u_{1},\ldots,u_{k}\}$. Define the assignment $A:V \to D$ by
  $A(u_{j})=a_{j}$ for $j = 1, \ldots, k$ and by $A(v) = z$ for all $v \in V \setminus
  U$. Notice that the $a_{j}$ may not necessarily be all distinct.

  For $d \in D$, let $\delta_{d}$ be the number of times $d$ appears in $(a_{1},
  \ldots, a_{k})$. Further define \[M= \prod_{d \in D, \delta_{d} \neq 0}
  \delta_{d}!\] 
  There are $M_{E} \leq M$ hyperedges $e$ in $E$ (including $(u_{1}, \ldots, u_{k})$) such that $P(A(e))=1$. Call these $e_{1}, \ldots, e_{M_{E}}$. Then
\[\Val_{H,P}(A) = \sum_{e \in E} w(e) P(A(e)) = \sum_{i=1}^{M_{E}}w(e_{i})>0.\]

Since $H_\eps$ is an $\eps$-$P$-sparsifier of $H$, at least one of $e_{1}, \ldots, e_{M_{E}}$ must be in $E_{\eps}$, since otherwise we would have 
\[\Val_{H_{\eps},P}(A) = \sum_{e \in E_{\eps}} w_{\eps}(e) P(A(e)) = 0 \notin (1
  \pm \eps) \Val_{H,P}(A).\]
Noticing that this argument holds for any hyperedge $f \in E$ and that $M \leq k!$, we have \[|E_{\eps}| \geq \frac{|E|}{M_{E}} \geq \frac{|E|}{M} \geq \frac{|E|}{k!}.\]
  Therefore, we have $|E_{\eps}| \geq |E|/k!$ and thus $|E_{\eps}| = \Omega(|E|)$.
\end{proof}

Notice that, if a $k$-ary predicate $P$ has an unused label, then $P$ contains a singleton $k$-cube. For singleton predicates with a very specific support (consisting of the same label), Proposition~\ref{prop:unused} is directly applicable.

\begin{proposition} \label{singleton:special}

Let $P:D^{k} \to \{0,1\}$ be a $k$-ary singleton predicate with $|D| \geq 2$
such that $P^{-1}(1) = \{(a,a,\ldots,a)\}$ for some $a \in D$. Then, for every
weighted directed $k$-uniform hypergraph $H=(V,E,w)$, for every $0<\eps<1$,
and for every partial subhypergraph $H_\eps=(V,E_\eps\subseteq E,w_\eps)$ of
$H$ which satisfies~(\ref{value2}), we have $|E_\eps|=\Omega(|E|)$.

\end{proposition}

\begin{proof} 
  By assumption, the support set of $P$ is not empty and $|D
  \setminus \{a\}| \geq 1$. Notice that, for any $z \in D \setminus \{a\}$, any
  $x_{1}, \ldots, x_{k-1} \in D$, and any permutation $\sigma$ on
  $\{1,2,\ldots,k\}$, we have 
  $P(\sigma(x_{1}, \ldots, x_{k-1},z)) = 0$.
  Thus, $z$ is an unused label and, by Proposition~\ref{prop:unused}, the claim
  follows.
\end{proof}

Since every singleton $k$-ary predicate $P$ contains a $k$-cube, by
Proposition~\ref{prop:singletoncube}, there exists a weighted directed
$k$-uniform hypergraph $H=(V,E,w)$ with $|V|=n$ such that for every $0<\eps<1$
and for every partial subhypergraph $H_{\eps}=(V,E_{\eps},w_{\eps})$ of $H$
which satisfies~(\ref{value2}), we have $|E_{\eps}| = \Omega(n^k)$. In
particular, the $k$-ary singleton predicate $\textsf{nOR}:D^k\to\{0,1\}$ has
this property, where \textsf{nOr} is defined by $\textsf{nOr}^{-1}(1) = \{(0,0, \ldots, 0)\}$. 
We use the concept of $k$-partite $k$-fold covers from
Appendix~\ref{sec:hypergraph} (and in particular Proposition~\ref{cutCover}) to
show that if any instance of $\CSPP$ has a (small) sparsifier then so does
$\CSP(\textsf{nOR})$, which establishes that singleton predicates cannot be sparsifiable.

\begin{theorem} \label{thm:singleton}
  Let $P:D^{k} \to \{0,1\}$ be a $k$-ary singleton predicate. If there is an
  $\eps$-$P$-sparsifier of size $g(n)$ then there is an
  $\eps$-\textsf{nOr}-sparsifier of size $O(g(n))$.
\end{theorem}

\begin{proof}
Without loss of generality, $D=[r]$ and $P^{-1}(1)=\{(a_{1}, \ldots, a_{k})\}$. 
Let $H=(V,E,w)$ be a weighted directed $k$-uniform hypergraph. 

We will show the existence of a function $f_{P}: Part_{r}(V) \to Part_{r}(V^{\gamma})$ such that
for any $\mathcal{A} \in Part_{r}(V)$ it holds that
\[\Val_{H,\textsf{nOr}}(\mathcal{A}) = \Val_{\gamma(H),P}(f_{P}(\mathcal{A})).\] 
The statement of the theorem then follows by Proposition~\ref{cutCover}.

Let $A^{P}_{j} = \bigcup_{i=1}^{k}A^{(i-1)}_{(j-a_{i})~(\textnormal{mod } r)}$. 
Define \[f_{P}(A_{0},\ldots,A_{r-1}) = (A^{P}_{0}, \ldots, A^{P}_{r-1}).\]
Moreover, define an assignment $A:V \to D$ by $A(v) = j \iff v \in A_{j}$. By definition, \[\Val_{H,\textsf{nOr}}(A) = \Val_{H,\textsf{nOr}}(A_{0}, \ldots, A_{r-1}).\]
Define the assignment $A^{\gamma}: V^{\gamma} \to D$ by $A^{\gamma}(v^{(i)}) = j
\iff v^{(i)} \in A^{P}_{j}$. We have \[\Val_{\gamma(H),P}(A^{\gamma})=
\Val_{\gamma(H),P}(A^{P}_{0}, \ldots,A^{P}_{r-1}) =
  \Val_{\gamma(H),P}(f_{P}(A_{0}, \ldots,A_{r-1})).\]
For a hyperedge $e = (v_{1}, \ldots, v_{k}) \in E$, define $\gamma(e)=e^{\gamma}=(v_{1}^{(0)}, \ldots, v_{k}^{(k-1)})$. We have
\begin{align*}
\textsf{nOr}(A(e))= 1 & \iff A(v_{1}) = \ldots = A(v_{k})=0\\
& \iff (v_{1}, \ldots, v_{k}) \in A_{0} \times \ldots \times A_{0}\\
& \iff \gamma((v_{1}, \ldots, v_{k}))=(v_{1}^{(0)}, \ldots, v_{k}^{(k-1)}) \in A_{0}^{(0)} \times \ldots \times A_{0}^{(k-1)}.
\end{align*}
Now, for $i=1, \ldots, k$,
\begin{alignat*}{3}
A_{0}^{(i-1)} \subseteq A^{P}_{j} & \iff A_{0}^{(i-1)} =
  A_{j-a_{i}~(\textnormal{mod }r)}^{(i-1)} && \quad \textnormal{ for some }j \in [r]\\
& \iff j = a_{i}~(\textnormal{mod } r) \iff j = a_{i} && \quad \textnormal{ (since } 0 \leq a_{i}, j < r).
\end{alignat*}
Therefore, assuming that $\textsf{nOr}(A(e))= 1$,
\begin{align*}
A_{0}^{(i-1)} \subseteq A^{P}_{j} & \implies (v_{1}^{(0)}, \ldots, v_{k}^{(k-1)}) \in A^{P}_{a_{1}} \times \ldots \times A^{P}_{a_{k}} \\
& \iff A^{\gamma}(v_{1}^{(0)}, \ldots, v_{k}^{(k-1)}) = (a_{1}, \ldots, a_{k}) \\
& \iff P(A^{\gamma}(v_{1}^{(0)}, \ldots, v_{k}^{(k-1)})) = P(A^{\gamma}(\gamma(v_{1}, \ldots, v_{k}))) = P(A^{\gamma}(\gamma(e)))= 1.
\end{align*}
Therefore, for any $e \in E$,
\[\textsf{nOr}(A(e)) = 1 \iff P(A^{\gamma}(\gamma(e)))=1\]
which implies
\begin{align*}
\Val_{H,\textsf{nOr}}(A_{0},\ldots, A_{r-1}) & = \Val_{H,\textsf{nOr}}(A) = \sum_{e \in E}w(e)\textsf{nOr}(A(e)) = \sum_{e \in E}w(e)P(A^{\gamma}(\gamma(e))) \\
& = \sum_{e \in E}w^{\gamma}(\gamma(e))P(A^{\gamma}(\gamma(e)))  = \sum_{e^{\gamma} \in E^{\gamma}}w^{\gamma}(e^{\gamma})P(A^{\gamma}(e^{\gamma})) \\
  & = \Val_{\gamma(H),P}(A^{\gamma}) = \Val_{\gamma(H),P}(f_{P}(A_{0},\ldots, A_{r-1})).
\end{align*}
\end{proof}

\subsection{Parity Predicates} \label{sec:parity}

The $k$-ary parity predicate $\Par:[r]^{k} \to \{0,1\}$ is defined by
\[\textsf{Par}(x_{1}, \ldots, x_{k})=1 \quad \iff \quad \sum_{i=1}^{k}x_{i} = 0 \textnormal{ (mod }2).\]

It is trivial to show that the parity predicates do not contain an unused label.
We will show that, for any $k \geq 3$, the $k$-ary parity predicate does
\emph{not} contain a singleton $\ell$-cube for any $\ell \leq k$, yet it
\emph{cannot} be written in terms of a hypergraph cut predicate. 

\begin{proposition}
For all $2 \leq \ell \leq k$, where $k\geq 3$, the $k$-ary parity predicate $\Par$ does not contain a singleton $\ell$-cube. 
\end{proposition}

\begin{proof} 
  By Definition~\ref{lcube}, the containment of a singleton $\ell$-cube for some
  $\ell\geq 3$ implies the containment of a singleton $2$-cube. Thus it suffices
  to show that $\Par$ does not contain any singleton $2$-cube.

  Suppose by contradiction that there exist subdomains $\{D_{j}=\{d^{j}_0,d^{j}_1\}\}_{j
  \in \{1,2\}} \in \binom{[r]}{2}$, indices $n_1,n_2 \in \{0,1\}$, and a
  permutation $\sigma$ on $\{1,2,\ldots,k\}$ such that there exist $x_{3}, \ldots, x_{k} \in [r] $ which satisfy
\[\Par(\sigma(d^{1}_{n_1}, d^{2}_{n_2}, x_{3}, \ldots, x_{k}))=1\]
and for all $y_{3}, \ldots, y_{k} \in [r]$, for all $i_{j} \in \{0,1\}$, 
\begin{equation} \label{singcubex} \Par(\sigma(d^{1}_{i_1}, d^{2}_{i_2}, y_{3}, \ldots, y_{k}))=1 \quad \implies \quad i_j=n_j \textnormal{ for all }j = 1, 2. \end{equation}

  \noindent\textbf{Case 1:} $d^{1}_0-d^{1}_1=0 \textnormal{ (mod $2$)}$. 
  
  Then, \[d^{1}_{n_1} +
  d^{2}_{n_2} + \sum_{j=3}^{k}x_{k} = d^{1}_{1-n_1} + d^{2}_{n_2} + \sum_{j=3}^{k}x_{k} \textnormal{ (mod $2$)}\]
and hence
\[\Par(\sigma(d^{1}_{1-n_1},d^{2}_{n_2},x_{3}, \ldots, x_{k})) =
  \Par(\sigma(d^{1}_{n_1},d^{2}_{n_2},x_{3}, \ldots, x_{k})) = 1,\]
contradicting~(\ref{singcubex}).

  \noindent\textbf{Case 2:} $d^{1}_0-d^{1}_1=1 \textnormal{ (mod $2$)}$. 
  
Then, \[d^{1}_{n_1} + d^{2}_{n_2} + \sum_{j=3}^{k}x_{k} = d^{1}_{1-n_1} + d^{2}_{n_2} + (x_{3}+1) + \sum_{j=4}^{k}x_{k} \textnormal{ (mod $2$)}\]
and hence
\[\Par(\sigma(d^{1}_{1-n_1},d^{2}_{n_2},x_{3}+1\textnormal{ (mod $2$)},x_{4}, \ldots, x_{k})) =
  \Par(\sigma(d^{1}_{n_1},d^{2}_{n_2},x_{3}, \ldots, x_{k})) = 1,\]
again contradicting~(\ref{singcubex}). 
\end{proof}

\begin{proposition}
  Let $\Par:[r]^k\to\{0,1\}$ be the $k$-ary parity predicate, where $k\geq 3$.
  Then, for all weighted directed $k$-uniform hypergraphs $H=(V,E,w)$ with $|V| \geq rk$, for all $r' \geq 2$, and for all functions $f:Part_{r}(V) \to Part_{r'}(V^{\gamma})$, there exists a partition of the vertices $\mathcal{A} \in Part_{r}(V)$ such that
  \[\Val_{H,\Par}(\mathcal{A}) \neq \Val_{\gamma(H), r'\textnormal{-}\NAE}(f(\mathcal{A})),\] where 
  $\gamma(H)$ is the $k$-partite $k$-fold cover of $H$.
\end{proposition}

\begin{proof}
We proceed by contradiction. Suppose that there exist a weighted directed
  $k$-uniform hypergraph $H=(V,E,w)$ with $|V| \geq rk$, an integer $r' \geq 2$,
  and a function $f_{\Par}:Part_{r}(V) \to Part_{r'}(V^{\gamma})$ such that for all partitions $\mathcal{A} \in Part_{r}(V)$ we have
\begin{equation} \label{parityAsCut}
  \Val_{H,\Par}(\mathcal{A}) = \Val_{\gamma(H),r'\textnormal{-}\NAE}(f(\mathcal{A})).
\end{equation}
  Let $A:V\to [r]$ be any assignment with the induced $r$-partition
  $\mathcal{A}=(A_{0}, \ldots, A_{r-1}) \in Part_{r}(V)$ such that $|A_{i}| \geq k$ for all $i \in [r]$. Denote $f_{\Par}(\mathcal{A})=(U_{0}, \ldots, U_{r'-1})$. 
Define an assignment $A_{f_{\Par}(\mathcal{A})}:V^{\gamma} \to [r']$ such that, for all $i \in [r']$ and for all $j \in [k]$,
\[A_{f_{\Par}(\mathcal{A})}(v^{(j)}) = i \iff A^{(j)}_{A(v)} \subseteq U_{i} .\]

First of all, we need to show that the assignment $A_{f_{\Par}(\mathcal{A})}$ is well-defined.
Notice that for all $i \in [r]$, for all $j \in [k]$, for all $u^{(j)},v^{(j)}
  \in A^{(j)}_{i}$ and for all $\ell \in [r']$ we must have $\{u^{(j)},v^{(j)}\}
  \cap U_{\ell} \in \{\emptyset, \{u^{(j)},v^{(j)}\}\}$. For suppose by
  contradiction that there exist $i \in [r]$, $j \in [k]$, and $u^{(j)},v^{(j)} \in A^{(j)}_{i}$ such that $u^{(j)} \in U_{\ell_{u}}$ and $v^{(j)} \in U_{\ell_{v}}$ with $\ell_{u} \neq \ell_{v}$. Assume without loss of generality that $j=0$. Then, for all $v_{2}, \ldots, v_{k} \in V$ and for $A_{i} \in \mathcal{A} \in Part_{r}(V)$, we would have 
\begin{align*}
\Par(A(u,v_{2} \ldots, v_{k})) & = \Par(A(u),A(v_{2}), \ldots, A(v_{k})) = \Par(i,A(v_{2}), \ldots, A(v_{k})) \\
& = \Par(A(v), A(v_{2}), \ldots, A(v_{k})) = \Par(A(v,v_{2}, \ldots, v_{k}))
\end{align*}
and hence
\begin{align*}
 r' \textnormal{-} \NAE(A_{f_{\Par}(\mathcal{A})}(u^{(0)},v_{2}^{(1)}, \ldots, v_{k}^{(k-1)})) & = \Par(A(u,v_{2}, \ldots, v_{k})) = \Par(A(v,v_{2}, \ldots, v_{k})) \\
 & =  r' \textnormal{-} \NAE(A_{f_{\Par}(\mathcal{A})}(v^{(0)},v_{2}^{(1)}, \ldots, v_{k}^{(k-1)})).
\end{align*}
Now pick $v_{2}, \ldots, v_{k}$ such that $\Par(A(u,v_{2}, \ldots, v_{k}))=0$. Then,
\begin{align} 
r' \textnormal{-} \NAE(A_{f_{\Par}(\mathcal{A})}(u^{(0)},v_{2}^{(1)}, \ldots, v_{k}^{(k-1)})) & = \Par(A(u,v_{2}, \ldots, v_{k})) = 0 \nonumber \\ 
& \implies v_{2}^{(1)}, \ldots, v_{k}^{(k-1)} \in U_{\ell_{u}} \label{ulu}
\end{align}
and
\begin{align}
r' \textnormal{-} \NAE(A_{f_{\Par}(\mathcal{A})}(v^{(0)},v_{2}^{(1)}, \ldots, v_{k}^{(k-1)})) & = \Par(A(v,v_{2}, \ldots, v_{k})) \nonumber \\ 
& = \Par(A(u,v_{2}, \ldots, v_{k})) = 0 \nonumber \\ 
& \implies v_{2}^{(1)}, \ldots, v_{k}^{(k-1)} \in U_{\ell_{v}}. \label{ulv}
\end{align}
Putting (\ref{ulu}) and (\ref{ulv}) together we get 
\begin{align*}
v_{2}^{(1)}, \ldots, v_{k}^{(k-1)} \in U_{\ell_{u}} \cap U_{\ell_{v}} & \implies  U_{\ell_{u}} \cap U_{\ell_{v}} \neq \emptyset \\
& \implies \ell_{u}=\ell_{v}
\end{align*}
contradicting our initial assumption that $\ell_{u} \neq \ell_{v}$. So for all $j \in [k]$, for all $u^{(j)},v^{(j)} \in A^{(j)}_{i}$ and for all $\ell \in [r']$ we have $\{u^{(j)},v^{(j)}\} \cap U_{\ell} \in \{\emptyset, \{u^{(j)},v^{(j)}\}\}$ and hence $A_{f_{\Par}(\mathcal{A})}$ is well-defined.

Now we want to consider vertices which belong to sets $A_{i}$ of different parity. Without loss of generality, pick $k$ vertices $u_{1}, \ldots, u_{k} \in A_{0}$ and $3$ vertices $v_{1}, v_{2}, v_{3} \in A_{1}$. Then we have
\begin{align*}
& \Par(A(v_{1}, u_{2}, \ldots, u_{k})) = \Par(1,0,\ldots,0)=0\\ \textnormal{ since } \quad & A(v_{1})+\sum_{j=2}^{k}A(u_{j}) = 1 \textnormal{ (mod }2)\end{align*}
and 
\begin{align*}
& \Par(A(v_{1}, v_{2}, v_{3}, u_{4}, \ldots, u_{k})) = \Par(1,1,1,0,\ldots,0)=0\\
\textnormal{ since } \quad & A(v_{1})+A(v_{2})+A(v_{3})+\sum_{j=4}^{k}A(u_{j}) = 3 = 1 \textnormal{ (mod }2).
\end{align*}
Then, by~(\ref{parityAsCut}) we must have
\[r'\textnormal{-}\NAE(A_{f_{\Par}(\mathcal{A})}(v_{1}^{(0)}, u_{2}^{(1)}, \ldots, u_{k}^{(k-1)}))=0\]
and
\[r'\textnormal{-}\NAE(A_{f_{\Par}(\mathcal{A})}(v_{1}^{(0)},v_{2}^{(1)},v_{3}^{(2)}, u_{4}^{(3)}, \ldots, u_{k}^{(k-1)}))=0\]
respectively.

By the definition of $r'$-$\NAE$, this implies that there exist $x,y \in [r']$ such that, for \[X=A_{1}^{(0)} \sqcup A_{0}^{(1)} \sqcup A_{0}^{(2)} \sqcup \ldots \sqcup A_{0}^{(k-1)} \] and \[Y=A_{1}^{(0)} \sqcup A_{1}^{(1)} \sqcup A_{1}^{(2)} \sqcup A_{0}^{(3)} \sqcup A_{0}^{(4)} \sqcup \ldots \sqcup A_{0}^{(k-1)}\] we have
\[X \cap U_{x} = X \quad \textnormal{ and } \quad Y \cap U_{y} = Y,\]
that is, hyperedges whose vertices lie wholly in $X$ or wholly in $Y$ do not contribute to the cut.
But then,
\[A_{1}^{(0)} \subseteq (X \cap Y) \subseteq  U_{x} \cap  U_{y}\]
which implies $U_{x} \cap  U_{y} \neq \emptyset$ and hence $x=y$.
It follows that 
\[A_{1}^{(0)} \sqcup A_{1}^{(1)} \sqcup A_{0}^{(2)} \sqcup A_{0}^{(3)} \sqcup \ldots \sqcup A_{0}^{(k-1)} \subseteq X \cup Y \subseteq U_{x}\]
and hence 
\[r'\textnormal{-}\NAE(A_{f_{\Par}(\mathcal{A})}(v_{1}^{(0)},v_{2}^{(1)},u_{3}^{(2)},
u_{4}^{(3)}, \ldots, u_{k}^{(k-1)}))=0\] implying, by~(\ref{parityAsCut}), that
\[\Par(A(v_{1}, v_{2}, u_{3}, u_{4}, \ldots, u_{k})) = \Par(1,1,0,0,\ldots,0)=0,\] a contradiction since 
\[A(v_{1})+A(v_{2})+\sum_{j=3}^{k}A(u_{j}) = 2 = 0 \textnormal{ (mod }2).\]
Therefore, such a map $f_{\Par}$ cannot exist.
\end{proof}

\end{document}